%% file: RationalRoughHestonLambdaFinal.tex
\documentclass[letter,11pt]{article}

\usepackage{a4wide}
\usepackage[english]{babel}

\usepackage{amsfonts}
\usepackage{amsmath}
\usepackage{graphicx}
\usepackage{caption}
\usepackage[colorlinks,bookmarks=true]{hyperref}
\usepackage{amsthm}
\usepackage{enumitem}
\usepackage{subfig}
\usepackage{bm}
\usepackage{bbm}

\numberwithin{equation}{section}
\numberwithin{figure}{section}

\newtheorem{lemma}{Lemma}[section]

\newtheorem{cor}{Corollary}[section]

\numberwithin{equation}{section}
\numberwithin{figure}{section}

\include{myCommands}

\title{A generalization of the rational rough Heston approximation}

\author{Jim Gatheral, Baruch College, CUNY,\\
{\tt  jim.gatheral@baruch.cuny.edu}\\$~~$\\
Rado\v{s} Radoi\v{c}i\'c, Baruch College, CUNY,\\ {\tt rados.radoicic@baruch.cuny.edu}
}

\begin{document}

\maketitle

\begin{abstract}

Previously, in \cite{gatheral2019rational}, we derived a rational approximation of the solution of the rough Heston fractional ODE in the special case $\lambda=0$, which corresponds to a pure power-law kernel.  In this paper we extend this solution to the general case of the Mittag-Leffler kernel with $\lambda \geq  0$.  We provide numerical evidence of the convergence of the solution.

\end{abstract}

\section{Introduction}\label{sec:intro}

In the case $\lambda \geq 0$, the rough Heston model of \cite{jaisson2016rough} may be written in forward variance form (see \cite{gatheral2019affine}) as
\bea
\frac{dS_t}{S_t} &=& \sqrt{V_t}\, \big\{   \rho\, dW_t + \sqrt{1-\rho^2}\,dW_t^{\perp} \big\}\nonumber\\
d\xi_t(u)&=&\sqrt{V_t}\,\kappa(u-t)  \,dW_t , \quad u\geq t
\label{eq:roughHeston}
\eea
where $\xi_t(u) = \eet{V_u}, u>t$ is the forward variance curve, $ \frac12 < \alpha  = H + \frac12\leq 1$, and the kernel $\kappa$ is given by
$$
\kappa(x) = \nu\,x^{\alpha-1}\,E_{\alpha,\alpha}(-\lambda\,x^\alpha),
$$
where $E_{\alpha,\alpha}(\cdot)$ denotes the generalized Mittag-Leffler function.

Let $X = \log S$ and $X_{t,T} := X_T-X_t$.
According to \cite{gatheral2019affine}, the forward variance model \eqref{eq:roughHeston} has a cumulant generating function (CGF) of the form
\beq
\varphi_t^T\left(a\right):= \log \eef{e^{\ui \,a\,X_{t,T}}}= \int_t^T\,\xi_t(s)\,g(T-s; a)\,ds
\label{eq:affineCGF}
\eeq
if and only if $g(t;u)$ satisfies the convolution integral equation
\beq
g = -\tfrac12 \,a\,(a+\ui) +\rho\,a\,\ui\,(\kappa \star g) +\tfrac 12 \, (\kappa \star g)^2,
\label{eq:Volterra}
\eeq
where $(\kappa \star g)(t;u) := \int_0^t\,\kappa(t - s)\,g(s;u)\,ds$.

Previously, in \cite{gatheral2019rational}, we derived various rational approximations to the solution of \eqref{eq:Volterra} in the special case $\lambda = 0$ where the kernel simplifies to 
\beq
\kappa_0(x) = \nu\,\frac{x^{\alpha-1}}{\Gamma(\alpha)}.
\label{eq:pLkernel}
\eeq
As pointed out in \cite{baschetti2022sinc} for example, such rational approximations are extremely fast to compute relative to the alternatives, enabling efficient calibration of the rough Heston model in this special case.

In the present note, we extend these rational approximations to the case $\lambda >0$. This enables fast calibration of the rough Heston model in the general case, with the extra parameter $\lambda$ providing additional flexibility to fit market implied volatility smiles.  Moreover, when $H=\tfrac12$, we retrieve the classical Heston model for which there is a well-known closed-form solution.

To proceed, let $D^\alpha$ and $I^{1-\alpha}$ represent respectively fractional differential and integral operators (see Appendix A of \cite{gatheral2019rational} for a very brief introduction to fractional calculus). 
The following result was originally proved in \cite{euch2019characteristic}.

\begin{lemma}\label{lem:11}
Let
$
\kappa(\tau) = \nu\,\tau^{\alpha-1}\,E_{\alpha,\alpha}(-\lambda\,\tau^\alpha)
$
and $
h(t;a) = \tfrac{1}{\nu}\,(\kappa\star g)(t;a)
$.   Then $h$ satisfies the fractional ODE
\beq
D^\alpha h(t;a) = -\frac12\,a\,(a+\ui) +(\ui\,\rho\,\nu\,a-\lambda) \, h(t;a) + \frac12\,\nu^2\, h^2(t;a); \quad I^{1-\alpha}h(t;a) = 0.
\label{eq:Riccati}
\eeq
\end{lemma}

\begin{proof}

For ease of notation, we drop the explicit dependence of $h$ and $g$ on $t$ and $a$.  Let $\kappa_0$ be the power-law kernel given by \eqref{eq:pLkernel}.
The Laplace transforms of $\kappa_0$ and $\kappa$ are given (for suitable $p$) by
\beas
\hat \kappa_0(p) = \frac{\nu}{p^\alpha}; 
\quad \hat \kappa(p) = \frac{\nu}{p^\alpha+\lambda}. 
\eeas
Let $\lambda' = \lambda/\nu$.  Then $ \hat \kappa_0 - \hat \kappa = \lambda'\, \hat \kappa_0 \,\hat \kappa$, and
$
\kappa_0 - \kappa  = \lambda'\,  \kappa_0 \star \kappa
$.
Also, by definition of the fractional integral operator,
\beas
I^\alpha f  = \frac1\nu\int_0^t\,\kappa_0(t-s)\,f(s)\,ds.
\eeas
Using that $( \kappa_0 \star \kappa) \star g =  \kappa_0 \star (\kappa \star g)$, it follows that 
\beas
h = \frac{1}{\nu}\, (\kappa\star g) =\frac{1}{\nu}\, \left( \kappa_0 \star g - \lambda\,\kappa_0 \star \kappa \star g \right)= I^\alpha g - \lambda\,I^\alpha h. 
\eeas
Operating with $D^\alpha$ gives
$
D^\alpha h(t;a)  = g(t;a) - \lambda\,h(t;a)
$.
Substituting into \eqref{eq:Volterra} gives the result.
\end{proof}

Given an approximate solution to the Riccati Equation \eqref{eq:Riccati}, an accurate approximation to the CGF \eqref{eq:affineCGF} may be easily computed.  European option prices may then be obtained using the Lewis formula (\cite{lewis2000option,gatheral2006volatility}):
\begin{equation}
C(S,K,T)=S-\sqrt{SK}\frac{1}{\pi}\int_0^\infty\frac{du}{u^2+\frac{1}{4}}
\,\mathrm{Re}\left[e^{-iuk}\varphi_t^T\left(u-i/2\right)\right]\label{eq:Lewis}
\end{equation}
where $S$ is the current stock price, $K$ the strike price and $T$ expiration.  Finally, implied volatilities may be computed by numerical inversion of the Black-Scholes formula.  

A key observation is that for option pricing with equation \eqref{eq:Lewis}, 
we need only find a good approximation to the solution $h(a,x)$ of the rough Heston Riccati equation for $a \in \cA$ with
\beq
\cA = \left\{z \in \CC: \Re(z)\geq 0, -1 \leq \Im(z) \leq 0 \right\}
\label{eq:cA}
\eeq
where $\Re$ and $\Im$ denote real and imaginary parts respectively.

\subsection{Main results and organization of the paper}

Our paper is organized as follows.  In Section \ref{sec:ShortTimes}, we derive a short-time expansion of the solution $h$ of the rough Heston Riccati equation \eqref{eq:Riccati}.  Then in Section \ref{sec:LongTimes}, we derive an asymptotic solution to \eqref{eq:Riccati} in the long-time limit $\tau = T-t \to \infty$.  In Section \ref{sec:Rational}, we explain how to construct global rational approximations to $h$ and present numerical results.  In particular, we exhibit (near) exponential convergence of the rational approximations with respect to their order.  Finally, in Section \ref{sec:summary}, we summarize and conclude.  Some technical details are relegated to the appendix. 


\section{Solving the rough Heston Riccati equation for short times}\label{sec:ShortTimes}

First, we derive a short-time expansion of the solution $h(t; a)$ of the fractional Riccati equation \eqref{eq:Riccati}.
Inspired by the $\lambda=0$ case, consider the small $t$ ansatz
\beq
h(t;a) = \sum_{j=1}^\infty\,b_j\,t^{j\,\alpha}.
\label{eq:smallt}
\eeq
Then, 
\beas
D^\alpha h = \sum_{j=1}^\infty\,b_j\,\frac{\Gamma(1+j\,\alpha)}{\Gamma(1+(j-1)\,\alpha)} \,t^{(j-1)\alpha}= \sum_{j=0}^\infty\,b_{j+1}\,\frac{\Gamma(1+(j+1)\,\alpha)}{\Gamma(1+j\,\alpha)} \,t^{j\, \alpha}.
\eeas
Substituting into \eqref{eq:Riccati} and matching coefficients of $t^0$ gives
\beas
b_1 = -\frac{1}{\Gamma(1+\alpha)}\,\frac12\,a(a+\ui) .
\eeas
Doing the same with the coefficient of 
$t^\alpha$ gives
\beas
b_2 =  \frac{\Gamma(1+ \alpha)}{\Gamma(1+ 2 \alpha)}\,( \ui\,\rho\,a-\lambda')\,\nu\,b_1,
\eeas
where as before, $\lambda' = \lambda/\nu$.
This generalizes to the recursion
\beas
b_1 &=& -\frac{1}{\Gamma(1+\alpha)}\,\frac12\,a(a+\ui)\nonumber\\
b_k &=& \frac{\Gamma(1+(k-1)\,\alpha)}{\Gamma(1+k\,\alpha)}\,\left\{-\tilde \lambda\,\nu\,b_{k-1} + \frac12\,\nu^2\,\sum_{i,j=1}^{k-1}\,\mathbbm{1}_{i+j = k-1}\,b_i\,b_j \right\},
\eeas
where $\tilde \lambda = \lambda' - \ui\,\rho\,a$.

\section{Solving the rough Heston Riccati equation for long times}\label{sec:LongTimes}

The fractional Riccati equation \eqref{eq:Riccati} can be re-expressed as
\beq
D^\alpha h(t;a) = \frac12\,\left(\nu\,h(t;a)-r_-\right)\,\left(\nu\,h(t;a)-r_+\right),
\label{eq:RiccatihTilde}
\eeq
with
$
A= \sqrt{a\,(a+i)+ ( \lambda'-\ui\,\rho\,a)^2 };\quad r_\pm = \left\{  \lambda' -\ui\,\rho\,a \pm A \right\}
$; $\lambda' = \lambda/\nu$.
Let
$
\nu\,h_{\infty}(t;a) = r_-\,\left[1-E_\alpha(-A\,\nu\,t^\alpha)\right]
$
where $E_\alpha$ is the Mittag-Leffler function.  
Then, for $t \in \RR_{\geq 0}$ and $a \in \cA$ (given in \eqref{eq:cA}), as in  Proposition 3.1 of \cite{gatheral2019rational},
$h_\infty(t;a)$ satisfies
\beq
\nu\,h_\infty (t;a) - r_- = -\frac{r_-}{A \nu}\,\frac{t^{-\alpha}}{\Gamma(1-\alpha)} + \cO\left(|A\,\nu\,t^\alpha|^{-2} \right).
\label{eq:hinfty}
\eeq
and thus
solves the rough Heston Riccati equation \eqref{eq:RiccatihTilde} up to an error term of $ \cO\left(|A\,\nu\,t^\alpha|^{-2} \right)$, as $ t  \to \infty$.

The form of the asymptotic expansion of $E_\alpha(-A\,\nu\,t^\alpha)$ in Corollary  \ref{cor:EalphaxInftyLev} motivates the following ansatz for $h(t;a)$ as $t \to \infty$:
\beq
h(t;a) =  \sum_{k=0}^\infty\,g_{k}\,t^{-k \alpha}.
\label{eq:bigt}
\eeq
Then
\bea
D^\alpha h(t;a)  
= \sum_{k=1}^\infty\,g_{k-1}\,\frac{\Gamma(1 - (k-1) \alpha)}{\Gamma(1 - k \,\alpha)}\,t^{-k \alpha}.
\label{eq:r1}
\eea
Note that \eqref{eq:hinfty} gives
\beas
g_0 = \frac{r_-}{\nu};\quad g_1 = -\frac{r_-}{A \nu^2}\,\frac{1}{\Gamma(1-\alpha)}.
\eeas
Also, from the fractional Riccati equation \eqref{eq:RiccatihTilde}, using that $g_0=r_-/\nu$,
\bea
D^\alpha h(a,x) &=& \frac{1}{2}\,(\nu\,h(t;a) - r_-)\,(\nu\,h(t;a)-r_+)\nonumber\\
&=&\nu\,\sum_{k=1}^\infty\,g_k\,t^{-k \alpha} \,\left(-A + \frac12\,\nu\,\sum_{k=1}^\infty\,g_k\,t^{-k \alpha}\right).
\label{eq:r2}
\eea
Equating \eqref{eq:r1} and \eqref{eq:r2} gives
$$
 \sum_{k=1}^\infty\,g_{k-1}\,\frac{\Gamma(1 - (k-1) \alpha)}{\Gamma(1 - k \,\alpha)}\,t^{-k \alpha} =
\nu\,\sum_{k=1}^\infty\,g_k\,t^{-k \alpha} \,\left(-A + \frac12\,\nu\,\sum_{k=1}^\infty\,g_k\,t^{-k \alpha}\right).
$$
Matching coefficients of $t^{-\alpha}$ gives
$$
g_{1}= -\frac{1}{A \nu}\,\frac{1}{\Gamma(1-\alpha)} \,g_0.
$$
Similarly, matching coefficients of $t^{-2 \alpha}$ gives
\beas
g_2= -\frac{1}{A\,\nu}\left\{\frac{\Gamma(1- \alpha)}{\Gamma(1- 2\,\alpha)}\, g_1 - \frac12\,\nu^2\, g_1^2\right\}.
\eeas
The general recursion for $k>2$ is given by
$$
g_{k} = -\frac{1}{A\,\nu}\left\{\frac{\Gamma(1-(k-1) \alpha)}{\Gamma(1-k \alpha)}\,g_{k-1} - \frac12\,\nu^2\,
\sum_{i,j=1}^{\infty}\,\mathbbm{1}_{i+j=k} \,g_i\,g_j\right\}.
$$

\section{Rational approximations of $h$ and numerical results}\label{sec:Rational}

In previous sections, we derived the short-time and long-time asymptotics of $h$.  
As in \cite{gatheral2019rational}, the only admissible global rational approximations of $h$ that match both \eqref{eq:smallt} and \eqref{eq:bigt} are of the diagonal form
\beq
h^{(n,n)}(t;a)=\frac{\sum_{i=1}^n\,p_{n,i} y^n}{\sum_{j=0}^n\,q_{n,j} y^n}
\label{eq:hnn}
\eeq
with $y = \nu\,t^\alpha$.

Explicit expressions for the coefficients $p_{n,i}$ and $q_{n,j} $ in \eqref{eq:hnn} are provided in the R-file {\verb roughHestonPadeLambda.R }, made openly accessible at  \url{https://github.com/jgatheral/RationalRoughHeston}, together with Jupyter notebooks illustrating the usage of the $h^{(n,n)}$.

\subsection{Numerical results}\label{sec:numerical}

Thanks to Giacomo Bormetti and his collaborators, we now have much more efficient code for the Adams scheme that seems to converge (for our purposes) with only 200 steps.  To be safe, our benchmark run of the Adams scheme will use 1,000 time steps.

In the following,
we assume the following realistic rough Heston parameters:
\beq
H = 0.05;\, \nu = 0.4; \rho=-0.65; \lambda = 1.
\label{eq:rHestonParams}
\eeq

The rational approximations $h^{(n,n)}(t;a)$ to $h(t;a)$ for $n \in \{2,3,4,5\}$ with the particular choice $a = 3-\ui/2$ and rough Heston parameters \eqref{eq:rHestonParams}  are plotted in Figure \ref{fig:hReImApprox}.  $h^{(3,3)}$, $h^{(4,4)}$ and $h^{(5,5)}$ are almost indistinguishable from the 1,000 step Adams estimate and significantly better than $h^{(2,2)}$. Moreover, all of these rational approximations are at least as fast to compute as the classical Heston solution.

\begin{figure} [h!]
\centering
\includegraphics[width=1.05 \linewidth ]{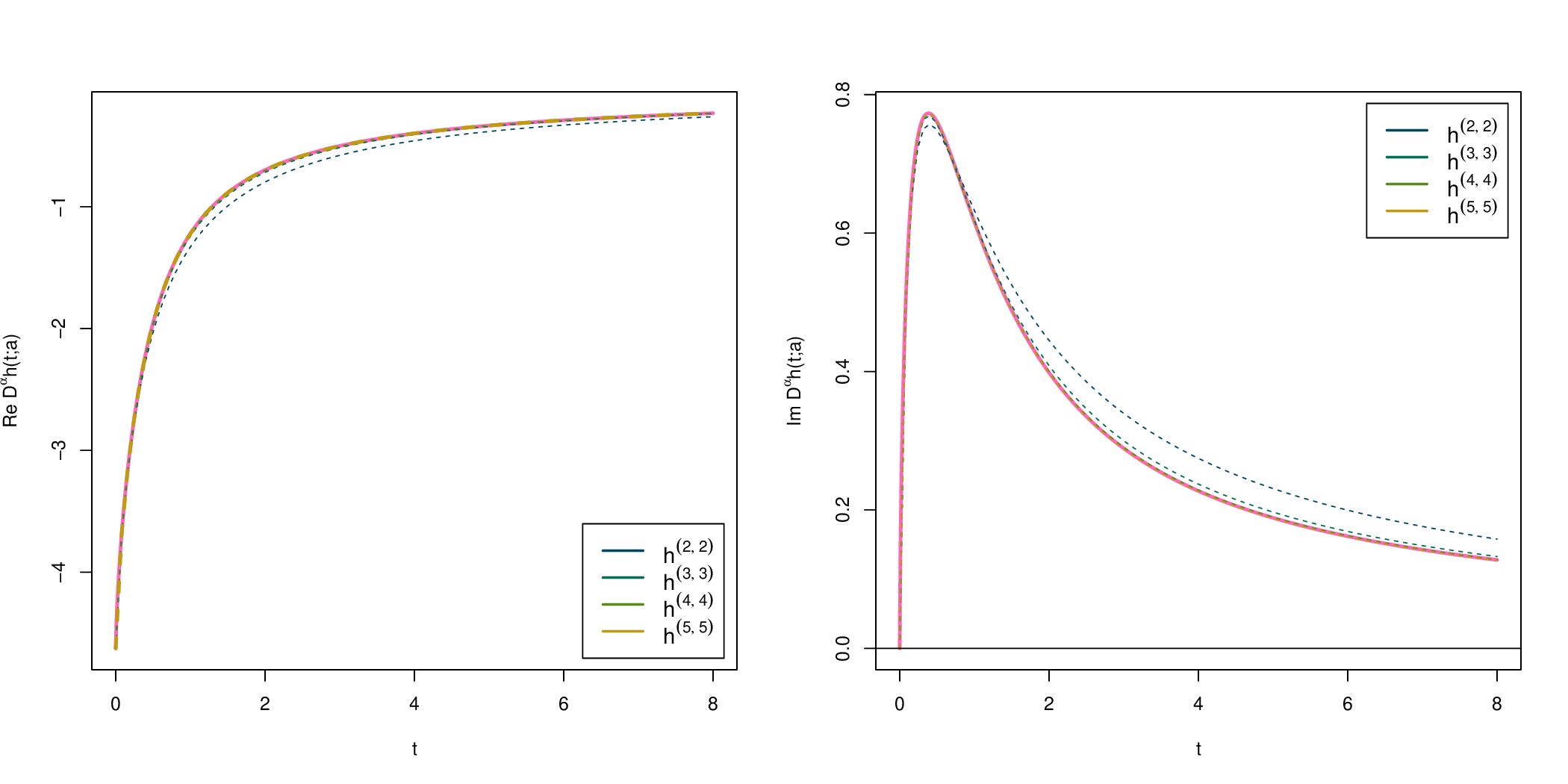}
\caption{In the left panel $\Re D^\alpha h(t;3-\ui/2)$ and in the right panel $\Im D^\alpha h(t; 3-\ui/2)$.  } 
\label{fig:hReImApprox}
\end{figure}

Naturally the coefficients of the higher order diagonal approximants become successively harder to compute in closed form.  And since the formulae are more complex, the functions when implemented take longer to compute.  Thus, even if, for example, $h^{(7,7)}$ were to be a better approximation than $h^{(4,4)}$, it would be much slower to compute and $h^{(4,4)}$ would likely be the approximation of choice in practice.

\subsection{Dependence of approximation quality on $H$}

So far, we have assessed the quality of our rational approximations with the realistic but fixed set of parameters \eqref{eq:rHestonParams}. 
It turns out that the quality of the rational approximations decreases as $H$ increases from $0$ to $\tfrac12$, which corresponds to the classical Heston model.
Indeed, it is evident from Figure \ref{fig:hReImApproxHvaries} that the  $h^{(n,n)}$ approximate $h$ almost perfectly when $H=0$;\footnote{More precisely in the limit $H \downarrow 0$ in the sense of \cite{forde2021small}.  } the approximation quality deteriorates as $H$ increases. Nevertheless, we observe that the approximation $h^{(5,5)}$ is very accurate, even in the classical case $H=\tfrac12$.

\begin{figure} [h!]
\centering
\includegraphics[width=.85\linewidth ]{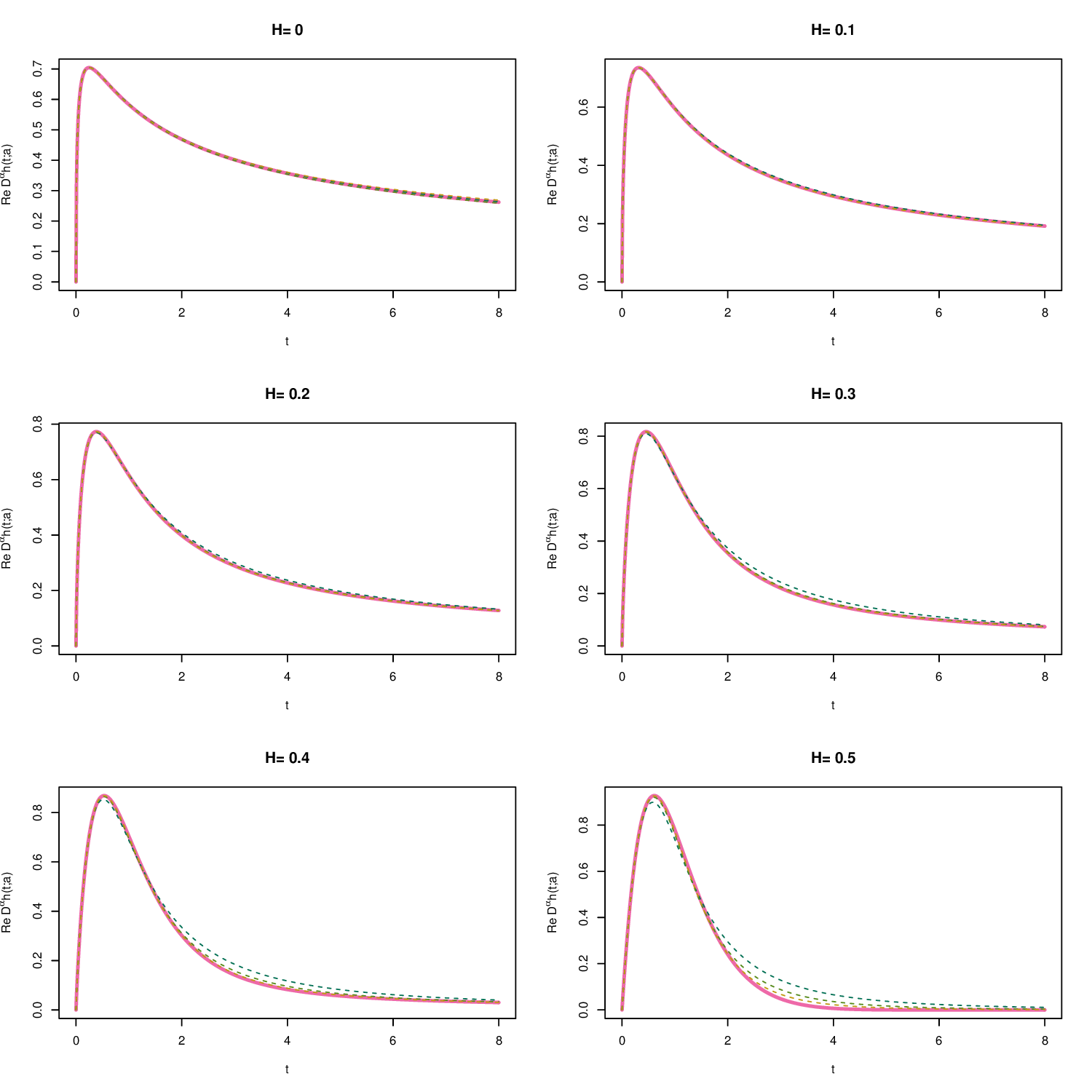}
\caption{$\Im D^\alpha h(t;3-\ui/2)$ computed for various values of $H$; solid pink lines are Adams scheme estimates with 1,000 steps; dashed lines are the rational approximations $h^{(3,3)}$, $h^{(4,4)}$, and $h^{(5,5)}$ respectively. The rational approximations and the numerical solution are almost indistinguishable when $H=0$. } 
\label{fig:hReImApproxHvaries}
\end{figure}

\subsection{Convergence of the $h^{(n,n)}$ in the classical Heston case}

The limit of the Mittag-Leffler kernel when $H = \tfrac 12$ is the exponential, corresponding to the classical Heston model.  Since we know the classical Heston characteristic function in closed form, we may study the convergence of the various rational approximations without numerical error from the Adams scheme.

In Figure \ref{fig:ClassicalHestonReIm}, we plot approximation errors in the classical Heston case for the Pad\'e approximations $h^{(2,2)}$ , $h^{(3,3)}$ , $h^{(4,4)}$ , $h^{(5,5)}$.  To the naked eye, it looks as if convergence of the rational approximations $h^{(n,n)}$ may be exponential in the approximation order $n$.  This conjecture is confirmed numerically in Figure \ref{fig:ClassicalHestonConvergence}.

\begin{figure} [h!]
\centering
\includegraphics[width=.9 \linewidth ]{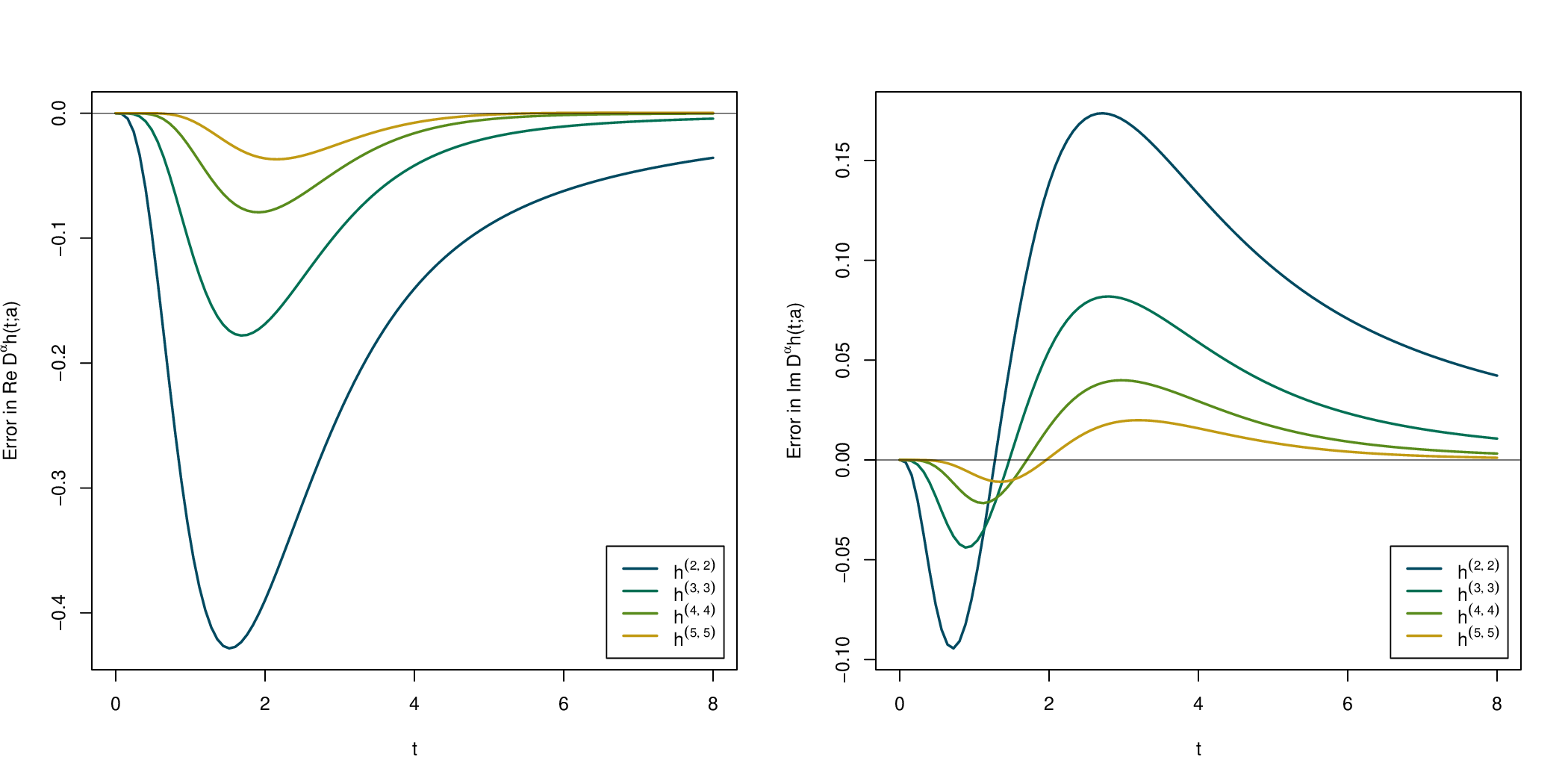}
\caption{Approximation errors for $\Re D^\alpha h(t;3-\ui/2)$ (left) and $\Im D^\alpha h(t;3-\ui/2)$ (right). }
\label{fig:ClassicalHestonReIm}
\end{figure}

\begin{figure} [h!]
\centering
\includegraphics[width=.85 \linewidth ]{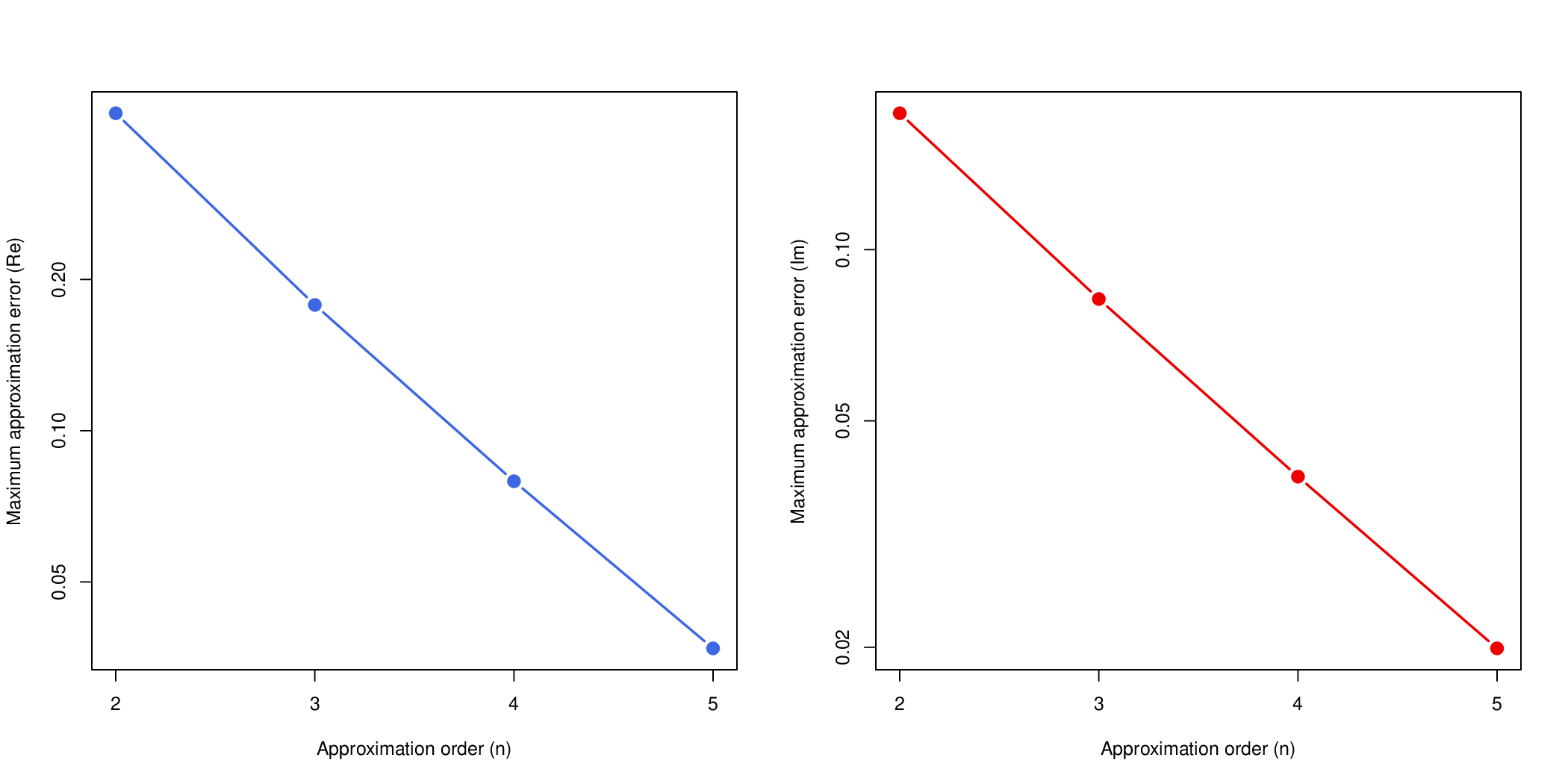}
\caption{Maximum absolute approximation errors for $\Re D^\alpha h(t;3-\ui/2)$ (left) and $\Im D^\alpha h(t;3-\ui/2)$ (right) as a function of the order $n$ of the Pad\'e approximation $h^{(n,n)}$.  The $y$-axis scale is logarithmic. }
\label{fig:ClassicalHestonConvergence}
\end{figure}


\subsection{Convergence in the  case $H < 1/2$}

In the general case $0 < H <1/2$, there is no closed-form solution for the characteristic function, so we must measure errors relative to the Adams scheme, which is itself an approximation.  
We choose an intermediate value $H=0.2$ for our experiments, which is high relative to the $0< H < 0.1$ typically estimated from time series or calibrated to implied volatilities. It is worth noting that the rational approximations are so good, that to get accurate estimates of approximation errors, the benchmark Adams scheme needs to be run with at least 1,000 steps.

\begin{figure} [h!]
\centering
\includegraphics[width=.9\linewidth ]{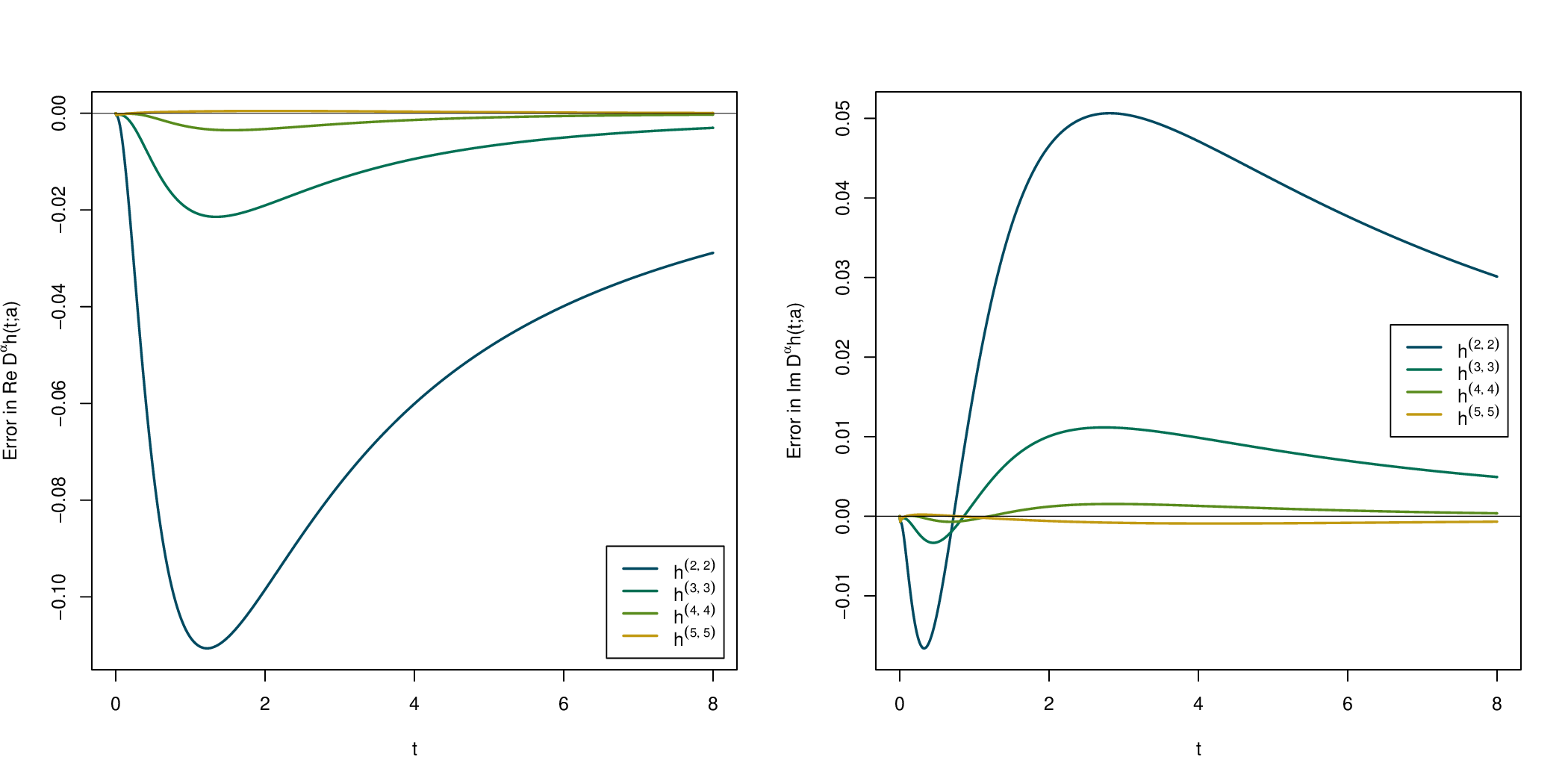}
\caption{Approximation errors for $\Re D^\alpha h(t;3-\ui/2)$ (left) and $\Im D^\alpha h(t;3-\ui/2)$ (right) in the case $H=0.2$. }
\label{fig:RoughlHestonHpt2ReIm}
\end{figure}

\begin{figure} [h!]
\centering
\includegraphics[width=.85\linewidth ]{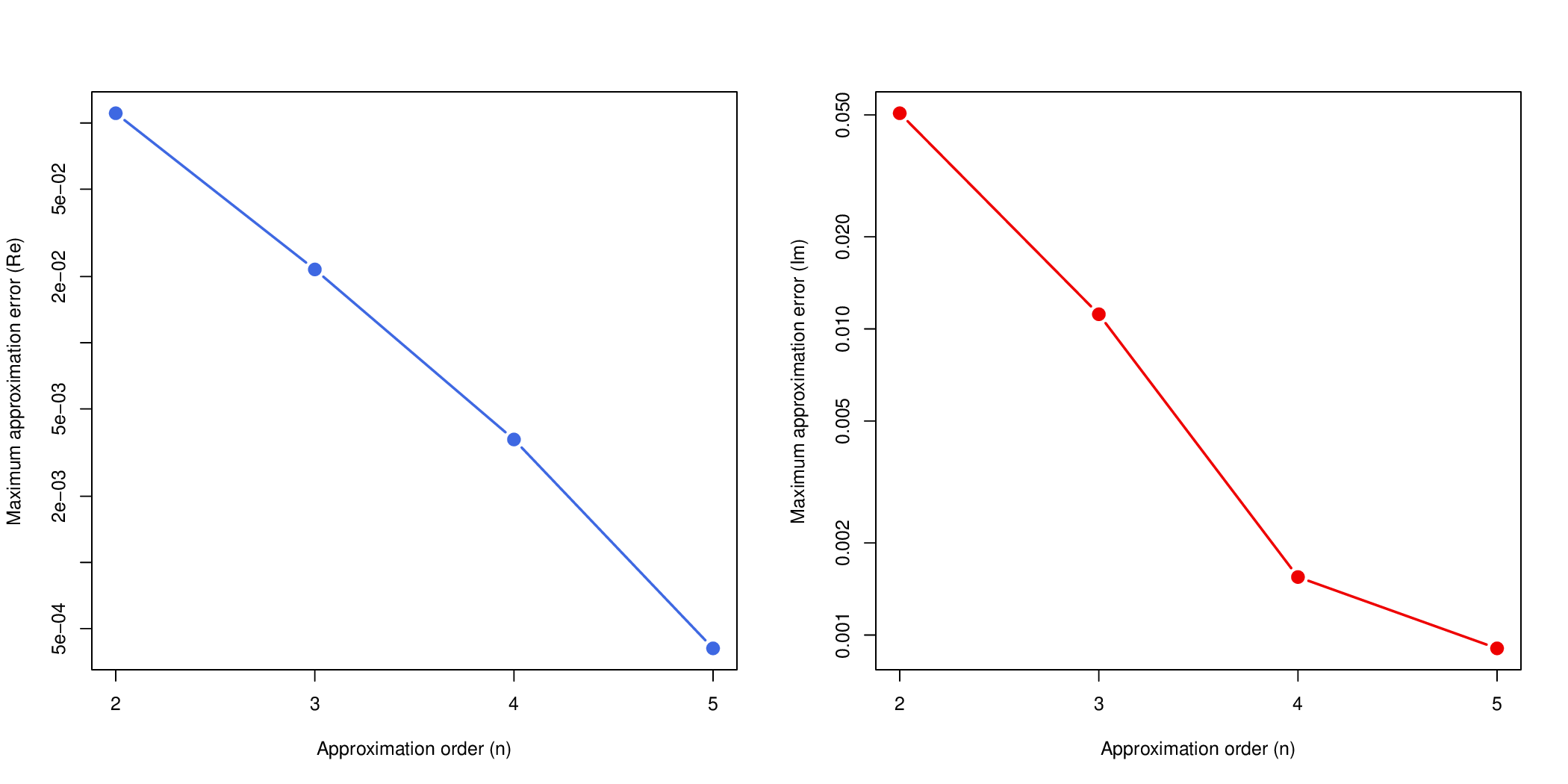}
\caption{Maximum absolute approximation errors for $\Re D^\alpha h(t;3-\ui/2)$ (left) and $\Im D^\alpha h(t;3-\ui/2)$ (right) as a function of the order $n$ of the Pad\'e approximation $h^{(n,n)}$.  The $y$-axis scale is logarithmic. }
\label{fig:RoughHestonHpt2Convergence}
\end{figure}

In Figure \ref{fig:RoughlHestonHpt2ReIm}, with $H=0.2$, we plot $h^{(n,n)}$ errors for $n \in \{2,3,4,5\}$.  To the naked eye, it looks as if convergence of the $h^{(n,n)}$  may once again be exponential in the approximation order $n$.  This conjecture is roughly confirmed in Figure \ref{fig:RoughHestonHpt2Convergence}.  Note also that relative to the classical Heston case, the sizes of the errors are much smaller, consistent with Figure \ref{fig:hReImApproxHvaries}.

\section{Summary and conclusions}\label{sec:summary}

The rough Heston cumulant generating function, as with all affine forward variance models, is a convolution of the forward variance curve with a function $g$ that satisfies the convolution Riccati integral equation \eqref{eq:Volterra}.  In \cite{gatheral2019rational} we constructed rational approximations to the solution of this equation in the special case where the kernel is power-law.  In this paper, we have extended that approximation to the more general case where the kernel is a Mittag-Leffler function. 

We focused on the diagonal approximants $h^{(n,n)}, n \in \{2,3,4,5\}$, the last three of which are efficient to compute, rendering them of great interest for practical application.  Moreover, we have provided numerical evidence of exponential convergence of the $h^{(n,n)}$ with respect to the approximation order $n$.  Code is made freely available online at  \url{https://github.com/jgatheral/RationalRoughHeston}.

\section{Acknowledgements}

We are grateful to Giacomo Bormetti and his collaborators for sharing their efficient Adams scheme code and to Stefano Marmi for enlightening conversations.

\bibliographystyle{alpha}
\bibliography{RoughVolatility}

\appendix

\section{Asymptotic expansion of the Mittag-Leffler function}
The following lemma is a straightforward corollary of Theorem 1.4 of  \cite{podlubny1998fractional}.

\begin{lemma}\label{lem:Pod1.4}
Let $0 < \alpha \leq 1$ and $\mu \in \RR$ be such that
$$
\frac{\pi \alpha}{2} < \mu <  \pi \alpha.
$$
Then, for any integer $p >0$, the following expansion holds:
$$
E_\alpha\left(z \right) = -\sum_{k=1}^p \frac{z^{-k}}{\Gamma(1-k\,\alpha)} + \cO\left(|z|^{-1-p}\right), \quad |z| \to \infty, \quad 
\mu \leq |\arg(z)| \leq \pi.
$$

\end{lemma}

\begin{lemma}\label{lem:lem2}
Let $0 < \alpha \leq 1$.  Further let $a = u + \ui \,y$ with $u \in \RR_{\geq 0}$, $y \in [-1,0]$ and let $A = \sqrt{a\,(a+\ui) + (\lambda'-\ui\,\rho\,a)^2},\, \nu>0,t>0$.
Then for any $x \in \RR_{>0}$, 
\beas
\left|\arg \left(-A \,x^\alpha \right)\right| \in  \left[\frac {3 \pi}4, \pi\right].
\eeas
\end{lemma}

\begin{proof}

Let $\bar \rho = \sqrt{1-\rho^2}$.  Then
\beas
\Re(A^2) 
= \bar\rho^2 \,u^2 -y\,(y+1)+(\lambda'+\rho\,y)^2,
\eeas
which is positive, so $\arg A^2 \in [-\frac \pi 2,\frac \pi 2]$, and so  $ \arg A \in [-\frac \pi 4,\frac \pi 4]$. It follows that 
\beas
\left|\arg \left(-A \,x^\alpha\right)\right| \in \left[\frac {3\,\pi} 4,\pi\right].
\eeas
\newline
\end{proof}

\begin{cor}\label{cor:EalphaxInftyLev}
Let $0 < \alpha \leq 1$.  Further let $a = u + \ui y$ with $u \in \RR_{>0}$ and $y \in [-1,0]$. Further let $A = \sqrt{a\,(a+\ui) + (\lambda'-\ui\,\rho\,a)^2},\, \nu>0,t>0$.  For any positive integer $p$ and $x \in \RR_{>0}$, 
\beas
E_\alpha\left(-A\,x^\alpha \right) = \sum_{k=1}^p \frac{(-1)^{k-1}}{A^k}\,\frac{x^{-k\,\alpha}}{\Gamma(1-k\,\alpha)} + \cO\left(\left|A\,x^\alpha\right|^{-1-p}\right), \quad x \to \infty.
\eeas

\end{cor}

\begin{proof}
Apply Lemma \ref{lem:Pod1.4} and Lemma \ref{lem:lem2} with $\mu = \frac34\,{\pi\,\alpha}$.
\end{proof}

\end{document}

%% file: myCommands.tex
\newcommand{\beas}{\begin{eqnarray*}}
\newcommand{\eeas}{\end{eqnarray*}}
\newcommand{\bal}{\begin{align}}
\newcommand{\eal}{\end{align}}
\newcommand{\bas}{\begin{align*}}
\newcommand{\eas}{\end{align*}}
\newcommand{\bea}{\begin{eqnarray}}
\newcommand{\eea}{\end{eqnarray}}
\newcommand{\tmop}{\end{eqnarray}}
\newcommand{\ben}{\begin{enumerate}}
\newcommand{\een}{\end{enumerate}}

\newcommand{\ui}{\mathrm{i}}

\newcommand{\cA}{\mathcal{A}}
\newcommand{\CC}{\mathbb{C}}

\newcommand{\cF}{\mathcal{F}}

\newcommand{\RR}{\mathbb{R}}

\newcommand{\cO}{\mathcal{O}}

\newcommand{\bi}{\begin{itemize}}
\newcommand{\ei}{\end{itemize}}
\newcommand{\beq}{\begin{equation}}
\newcommand{\eeq}{\end{equation}}
\newcommand{\bv}{\begin{verbatim}}
\newcommand{\ev}{\end{verbatim}}

\newcommand{\eef}[1]{\ensuremath{\mathbb{E}\left[\left.{#1}\right|\cF_t\right]}}

\newcommand{\eet}[1]{\mathbb{E}_{t}\left[#1\right]}

%% file: RationalRoughHestonLambdaFinal.bbl
\begin{thebibliography}{BBRR22}

\bibitem[BBRR22]{baschetti2022sinc}
Fabio Baschetti, Giacomo Bormetti, Silvia Romagnoli, and Pietro Rossi.
\newblock The {SINC} way: A fast and accurate approach to {F}ourier pricing.
\newblock {\em {Quantitative Finance}}, 22(3):427--446, 2022.

\bibitem[EER19]{euch2019characteristic}
Omar El~Euch and Mathieu Rosenbaum.
\newblock The characteristic function of rough {H}eston models.
\newblock {\em Mathematical Finance}, 29(1):3--38, 2019.

\bibitem[FGS21]{forde2021small}
Martin Forde, Stefan Gerhold, and Benjamin Smith.
\newblock Small-time, large-time, and asymptotics for the {Rough Heston} model.
\newblock {\em Mathematical Finance}, 31(1):203--241, 2021.

\bibitem[Gat06]{gatheral2006volatility}
Jim Gatheral.
\newblock {\em The volatility surface: {A} practitioner's guide}.
\newblock John Wiley \& Sons, 2006.

\bibitem[GKR19]{gatheral2019affine}
Jim Gatheral and Martin Keller-Ressel.
\newblock Affine forward variance models.
\newblock {\em Finance and Stochastics}, 23(3):501--533, 2019.

\bibitem[GR19]{gatheral2019rational}
Jim Gatheral and Rado{\v s} Radoi{\v c}i\'c.
\newblock {Rational approximation of the rough Heston solution}.
\newblock {\em International Journal of Theoretical and Applied Finance},
  22(3):1950010, 2019.

\bibitem[JR16]{jaisson2016rough}
Thibault Jaisson and Mathieu Rosenbaum.
\newblock Rough fractional diffusions as scaling limits of nearly unstable
  heavy tailed {H}awkes processes.
\newblock {\em The Annals of Applied Probability}, 26(5):2860--2882, 2016.

\bibitem[Lew00]{lewis2000option}
AL~Lewis.
\newblock {\em Option Valuation under Stochastic Volatility}.
\newblock Finance Press: Newport Beach, CA, 2000.

\bibitem[Pod98]{podlubny1998fractional}
Igor Podlubny.
\newblock {\em Fractional differential equations: an introduction to fractional
  derivatives, fractional differential equations, to methods of their solution
  and some of their applications}, volume 198.
\newblock Elsevier, 1998.

\end{thebibliography}
